\documentclass[acmsmall]{acmart}

\usepackage[utf8]{inputenc} 
\usepackage[T1]{fontenc}    
\usepackage{textcomp}
\usepackage{hyperref}       
\usepackage{url}            
\usepackage{booktabs}       
\usepackage{amsfonts}       
\usepackage{nicefrac}       
\usepackage{microtype}      

\usepackage[ruled, algo2e, vlined]{algorithm2e}
\usepackage{setspace}

\usepackage{epsfig,bm,subfigure,graphicx,color}
\usepackage{subfigure}
\usepackage{wrapfig}

\usepackage{amsmath,amsfonts,amsthm,mathtools}
\usepackage{bbm}
\usepackage{lscape}
\usepackage[normalem]{ulem}
\usepackage{tipa,tipx}

\usepackage{tcolorbox}

\newtheorem{theorem}{Theorem}

\theoremstyle{definition}

\usepackage{multirow}
\usepackage{tabularx}

\usepackage{pifont}
\newcommand{\cmark}{\ding{51}}%
\newcommand{\xmark}{\ding{55}}%

\AtBeginDocument{%
  \providecommand\BibTeX{{%
    \normalfont B\kern-0.5em{\scshape i\kern-0.25em b}\kern-0.8em\TeX}}}

\acmConference[Woodstock '18]{Woodstock '18: ACM Symposium on Neural
  Gaze Detection}{June 03--05, 2018}{Woodstock, NY}
\acmBooktitle{Woodstock '18: ACM Symposium on Neural Gaze Detection,
  June 03--05, 2018, Woodstock, NY}
\acmPrice{15.00}
\acmISBN{978-1-4503-XXXX-X/18/06}

\begin{document}

\title{Overhead-free User-side Recommender Systems}

\author{Ryoma Sato}
\email{rsato@nii.ac.jp}
\affiliation{%
  \institution{National Institute of Informatics}
  \country{Japan}
}


\renewcommand{\shortauthors}{Sato et al.}


\begin{abstract}
    Traditionally, recommendation algorithms have been designed for service developers. But recently, a new paradigm called user-side recommender systems has been proposed. User-side recommender systems are built and used by end users, in sharp contrast to traditional provider-side recommender systems. Even if the official recommender system offered by the provider is not fair, end users can create and enjoy their own user-side recommender systems by themselves. Although the concept of user-side recommender systems is attractive, the problem is they require tremendous communication costs between the user and the official system. Even the most efficient user-side recommender systems require about $5 \times$ more costs than provider-side recommender systems. Such high costs hinder the adoption of user-side recommender systems. In this paper, we propose overhead-free user-side recommender systems, \textsc{RecCycle}, which realizes user-side recommender systems without any communication overhead. The main idea of \textsc{RecCycle} is to recycle past recommendation results offered by the provider's recommender systems. The ingredients of \textsc{RecCycle} can be retrieved ``for free,'' and it greatly reduces the cost of user-side recommendations. In the experiments, we confirm that \textsc{RecCycle} performs as well as state-of-the-art user-side recommendation algorithms while \textsc{RecCycle} reduces costs significantly.
\end{abstract}



\maketitle

\section{Introduction}

Recommender systems have been used in many web services \cite{linden2003amazon, geyik2019fairness}. It was estimated that $35$ \% of purchases on Amazon and $75$ \% of watches on Netflix came from recommender systems \cite{mackenzie2013how}. Recommender systems are indispensable both for businesses and users.

Although traditional recommender systems aim only at conversion, many fine-grained demands for recommender systems have emerged. Users may want to receive fair recommendations \cite{kamishima2012enhancement, biega2018equity, milano2020recommender} or serendipitous recommendations \cite{chen2021values, anderson2020algorithmic, steck2018calibrated, mladenov2020optimizing, zheng2018fairness}, or users may want recommender systems to be transparent \cite{sinha2002role, balog2019transparent} and steerable \cite{green2009generating, balog2019transparent}. For example, on LinkedIn, recruiters may want to receive account recommendations that are fair in terms of gender and race to avoid (implicit) discrimination. A citizen who gathers information for election may want to receive both Republican and Democrat news equitably to avoid filter bubbles \cite{pariser2011filter}. Cinema enthusiasts may want to receive recommendations that involve minor movies instead of popular movies that enthusiasts already know.

However, there are too many kinds of demands, and the service provider cannot cope with all of them. Besides, service provider may not implement such functionalities on purpuse. For example, some service providers may intentionally choose to increase short-term conversions instead of caring the fairness of the platform.

If the service provider does not implement fair recommender systems, users are forced to use unfair ones or quit the service. It has been considered that users have little ability to change the recommendations. In most cases, the only option available to the user is to wait until the service implements the functionality. \citet{green2009generating} also pointed out that ``If users are unsatisfied with the recommendations generated by a particular system, often their only way to change how recommendations are generated in the future is to provide thumbs-up or thumbs-down ratings to the system.''

User-side recommender systems \cite{sato2022private} offer a proactive solution to this problem. Users can build their own (i.e., private, personal, or user-side) recommender systems to ensure recommendations are made fairly and transparently. Since the system is built by the user, it can be customized to meet specific criteria they want and add the functionalities they want. User-side recommender systems realize ultimate personalization.

\begin{table}[tb]
    \small
    \caption{Properties of user-side recommender systems. The definitions of these properties are shown in Section \ref{sec: properties}. Postprocessing (PP) applies postprocessing directly to the official recommender system, which is not sound when the list does not contain some sensitive groups (See also Section \ref{sec: challenge}).}
    \vspace{-0.1in}
    \centering
    \begin{tabular}{lcccccc} \toprule
                      & PP     & \textsc{PrivateRank} \cite{sato2022private} & \textsc{PrivateWalk} \cite{sato2022private} & ETP \cite{sato2022towards} & \textsc{Consul} \cite{sato2022towards} & \textsc{RecCycle} (ours) \\ \midrule
        Consistent    & {\color[HTML]{03af7a} \cmark} & {\color[HTML]{03af7a} \cmark}                                     & {\color[HTML]{ff4b00} \xmark}                                     & {\color[HTML]{03af7a} \cmark}                    & {\color[HTML]{03af7a} \cmark}                                & {\color[HTML]{03af7a} \cmark}                     \\
        Sound         & {\color[HTML]{ff4b00} \xmark} & {\color[HTML]{03af7a} \cmark}                                     & {\color[HTML]{03af7a} \cmark}                                     & {\color[HTML]{03af7a} \cmark}                    & {\color[HTML]{03af7a} \cmark}                                & {\color[HTML]{03af7a} \cmark}                     \\
        Local         & {\color[HTML]{03af7a} \cmark} & {\color[HTML]{ff4b00} \xmark}                                     & {\color[HTML]{03af7a} \cmark}                                     & {\color[HTML]{ff4b00} \xmark}                    & {\color[HTML]{03af7a} \cmark}                                & {\color[HTML]{03af7a} \cmark}                     \\
        Overhead-free & {\color[HTML]{03af7a} \cmark} & {\color[HTML]{ff4b00} \xmark}                                     & {\color[HTML]{ff4b00} \xmark}                                     & {\color[HTML]{ff4b00} \xmark}                    & {\color[HTML]{ff4b00} \xmark}                                & {\color[HTML]{03af7a} \cmark}                     \\ \bottomrule
    \end{tabular}
    \label{tab: prop}
\end{table}

The concept of user-side recommender systems looks similar to steerable \cite{green2009generating} (or scrutable \cite{balog2019transparent}) recommender systems at first glance. Steerable recommender systems also allow users to control the recommendation results. However, the key difference is that steerable systems are implemented by the service provider, while user-side recommender systems are built by the users themselves. What to steel is chosen by the service provider in traditional steerable recommender systems. If the recommender system in use is not steerable in the way they want, users cannot enjoy steerability and must wait for the service provider to implement it. By contrast, user-side recommender systems allow users to make the system steerable, even if the service provider implemented only a standard non-steerable system.

Although user-side recommender systems are attractive, building them is challenging. End users do not have access to the data stored in the service's database, unlike the developers employed by the service provider. Most modern recommender systems rely on user log data and/or item features to make recommendations. At first glance, it seems impossible to build an effective recommender system without such data. \citet{sato2022private} addressed this problem by using the official recommender systems provided by the target web service. Although the official recommender systems are black-box and possibly unfair, Sato's methods turn them into fair and transparent ones on the user's side by combining multiple outputs. However, existing user-side recommender systems issue multiple queries to the official (possibly unfair) recommender system to build a single (fair) recommendation list. In other words, these methods trade communication costs with fairness. The drawback of this approach is the communication cost. Even the most efficient user-side recommender systems, \textsc{Consul}, require $5$ queries to build a recommendation list \cite{sato2022towards}. This means that \textsc{Consul} loads the service $5$ times more. Such a high communication cost causes problems. First, the service provider may disfavor and prohibit such users' activities to mitigate the load on the service. Second, end users cannot afford to pay the high API cost. Third, such systems are not suitable for real-time applications due to the response time of the multiple queries.

We advocate that the communication cost between the end user and the service is crucial for effective user-side recommender systems. An ideal user-side system works as if it were an official system. The recommendation list should be shown to the user at the same time as the official system. However, existing user-side recommender systems require additional queries and thus require more loading time than the official system, which leads to a poor user experience.

We propose overhead-free user-side recommender systems, \textsc{RecCycle} (recommendation + recycle), to address this problem. The main idea of \textsc{RecCycle} is to recycle past recommendation results presented by the provider's recommender systems when the user uses the system as usual. These recommendation results used to be discarded once shown on the page. Sometimes, these recommendations are just shown on the page and do not catch the attention of the user due to the position and/or timing of the presentation. \textsc{RecCycle} ``recycles'' these information to create new recommendations on the user's side. These information can be used ``for free'', i.e., without any additional communication cost. All of the computation for \textsc{RecCycle} is done locally. \textsc{RecCycle} is so communication efficient that it can realize real-time user-side recommendations, and the user can enjoy the recommendations as if they were shown by the official system. 

\textsc{RecCycle} can be combined with existing user-side recommender systems. We will elaborate on the premise of \textsc{RecCycle} in the following sections. As a special case, we show that \textsc{RecCycle} can be combined with \textsc{Consul} \cite{sato2022towards}, which leads to consistent, sound, local, and overhead-free user-side recommender systems (Table \ref{tab: prop}).

In the experiments, we confirm that \textsc{RecCycle} performs as well as state-of-the-art user-side recommendation algorithms while \textsc{RecCycle} reduces costs significantly.

The contributions of this study are as follows:
\begin{itemize}
    \item We propose overhead-free user-side recommender systems, \textsc{RecCycle}, for the first time.
    \item We show that \textsc{RecCycle} is consistent, sound, and local, as well as overhead-free.
    \item We empirically validate that \textsc{RecCycle} performs as well as state-of-the-art user-side recommendation algorithms while \textsc{RecCycle} reduces costs significantly.
    \item We deploy \textsc{RecCycle} in a real-world X (Twitter) environment and confirm that users can realize their own recommender system with specified functionalities they call for using \textsc{RecCycle}.
\end{itemize}

\section{Notations}

For every positive integer $n \in \mathbb{Z}_+$, $[n]$ denotes the set $\{ 1, 2, \dots n \}$.
Let $\mathcal{I} = [n]$ denote the set of items, where $n$ is the number of items. Without loss of generality, we assume that the items are numbered with $1, \dots, n$. $K \in \mathbb{Z}_+$ denotes the length of a recommendation list. The notations are summarized in Table \ref{tab: notations}.

\section{Problem Setting} \label{sec: setting}

\begin{table}[tb]
    \centering
    \caption{Notations.}
    \vspace{-0.1in}
    \begin{tabular}{ll} \toprule
        Notations                         & Descriptions                                         \\ \midrule
        $[n]$                             & The set $\{ 1, 2, \dots, n \}$.                      \\
        $G = (V, E)$                      & A graph.                                             \\
        $\mathcal{I} = [n]$               & The set of items.                                    \\
        $\mathcal{A}$                     & The set of protected groups.                         \\
        $a_i \in \mathcal{A}$             & The protected attribute of item $i \in \mathcal{I}$. \\
        $\mathcal{H} \subset \mathcal{I}$ & The set of items that have been interacted with.     \\
        $\mathcal{P}_{\text{prov}}$       & The provider's official recommender system.          \\
        $K \in \mathbb{Z}_+$              & The length of a recommendation list.                 \\
        $\tau \in \mathbb{Z}_{\ge 0}$     & The minimal requirement of fairness.                 \\
        \bottomrule
    \end{tabular}
    \label{tab: notations}
    \vspace{-0.1in}
\end{table}

We follow the basic setting of user-side recommender systems \cite{sato2022private, sato2022towards}. Suppose we are an end user of the service (e.g., Twitter). The goal is to build our own recommender system without accessing privileged information stored in the service's database.

\subsection{Provider Recommender System}

In this paper, we focus on item-to-item recommendations, following \citet{sato2022private, sato2022towards}\footnote{Note that item-to-user recommender systems can also be built based on item-to-item recommender systems by, e.g., gathering items recommended by an item-to-item recommender system for the items that the user has purchased before.}. Specifically, when we visit the page associated with item $i$. The official recommender system of the service presents $K$ items $\mathcal{P}_\text{prov}(i) \in \mathcal{I}^K$, i.e., for $k = 1, 2, \cdots, K$, $\mathcal{P}_\text{prov}(i)_k \in \mathcal{I}$ is the $k$-th relevant item to item $i$ according to the service's recommender system. For example, $\mathcal{P}_\text{prov}$ is observed in the ``Customers who liked this also liked'' panel in e-commerce platforms. We call $\mathcal{P}_\text{prov}$ the \emph{service provider's official recommender system}. We assume that $\mathcal{P}_\text{prov}$ provides relevant items but is unfair and is a black-box system. The goal is to build a fair and white-box recommender system by leveraging the provider's recommender system.

\subsection{Sensitive Attributes}

We encode the functionalities we call for by sensitive attributes following \cite{sato2022private, sato2022towards}. We assume that each item $i$ has a discrete sensitive attribute $a_i \in \mathcal{A}$, where $\mathcal{A}$ is the set of sensitive groups. For example, in a talent market service, each item represents a person, and $\mathcal{A}$ can be gender or race. In a news recommender system, each item represents a news article, and $\mathcal{A}$ can be $\{$Republican, Democrat$\}$. In news recommender systems,
$\mathcal{A}$ can also be $\{$ $\ge$ 10\,000 views, $\ge$ 1\,000 \& < 10\,000 views, < 1\,000 views $\}$, which leads to a recommender system that includes minor news articles as well. Many functionalities, including fairness, diversity, and serendipity, can be realized by leveraging the sensitive attributes. What sensitive attributes to use is up to the user's demand, and users can choose the sensitive attributes they want to use. In other words, each user can customize their own recommender system by choosing the sensitive attribute. In the following, we generalize the sensitive attributes as $a_i \in \mathcal{A}$, and the user can substitute $\mathcal{A}$ with the sensitive attributes they want to use.

We want a user-side recommender system that offers items from each group in a certain proportion, e.g., so that each group is shown equally, or demographic parity holds. The proportion can also be specified by the user. We assume that sensitive attribute $a_i$ \emph{can be observed}, which is the common assumption in \cite{sato2022private, sato2022towards}. Admittedly, this assumption does not necessarily hold in practice. However, when this assumption is violated, one can estimate $a_i$ from auxiliary information, and the estimation of the attribute is an ordinary supervised learning task and can be solved by off-the-shelf methods, such as neural networks and random forests. As the attribute estimation process is not relevant to the core of user-side recommender system algorithms, this study focuses on the setting where the true $a_i$ can be observed.

\subsection{Communication Cost}

Existing user-side recommender systems \cite{sato2022private, sato2022towards} incur high communication costs. When the user accesses the page associated with item $i$, these methods retrieve recommendations $\mathcal{P}_\text{prov}(i_1), \ldots, \mathcal{P}_\text{prov}(i_c)$ for some items $i_1, \ldots, i_c$ and make a fair recommendation for item $i$ by combining these lists. These ingredients are retrieved by directly calling the API if available or crawling the service site otherwise. It means that these methods require downloading $c$ item pages to build a single recommendation list and incur $c$ times more communication costs than the official recommender system. Such a high communication cost leads to a long loading time, and it becomes impossible to obtain recommendation results in real time. Worse, such methods significantly increase the load on the service, and the service provider may prohibit such activities. In this paper, we aim to build a user-side recommender system that realizes fair recommendations with little or no communication overhead. 

\subsection{Problem Setting}

The problem setting can be summarized as follows:

\begin{tcolorbox}[colframe=gray!20,colback=gray!20,sharp corners]
    \noindent \uline{\textbf{User-side Recommender System Problem.}}\\
    \textbf{Given:} Oracle access to the official recommendations $\mathcal{P}_{\text{prov}}$. Sensitive attribute $a_i \in \mathcal{A}$ of each item $i \in \mathcal{I}$.\\
    \textbf{Output:} A user-side recommender system $\mathcal{Q}\colon \mathcal{I} \to \mathcal{I}^K$ that is fair with respect to $\mathcal{A}$.\\
    $\mathcal{Q}$ should be obtained with as few evaluations of $\mathcal{P}_{\text{prov}}$ as possible.
\end{tcolorbox}

\section{Proposed Method}

We introduce our proposed method, \textsc{RecCycle}. The main idea of \textsc{RecCycle} is to store the past recommendations shown by the provider's recommender systems and utilize them to create new recommendations on the user's side. However, it is not straightforward to use the past recommendations directly as we will see in the following because the past recommendations are not fair.

\subsection{Challenge} \label{sec: challenge}

Simple postprocessing of the past recommendations fails. Suppose we want to build an account recommender system on X (Twitter). We set $\mathcal{A} = \{${\color[HTML]{990099} man}, {\color[HTML]{F69900} woman}$\}$. We visit the page of $i =$ Tom Hanks, and the official recommender system shows $K = 6$ items as follows: \begin{align}
    \mathcal{P}_\text{prov}(i = \text{Tom Hanks})_1 & = \text{Seth Macfarlane ({\color[HTML]{990099} man})} \\
    \mathcal{P}_\text{prov}(i = \text{Tom Hanks})_2 & = \text{Danny Devito ({\color[HTML]{990099} man})}     \\
    \mathcal{P}_\text{prov}(i = \text{Tom Hanks})_3 & = \text{Leonardo Dicaprio ({\color[HTML]{990099} man})}    \\
    \mathcal{P}_\text{prov}(i = \text{Tom Hanks})_4 & = \text{Jason Bateman ({\color[HTML]{990099} man})}       \\
    \mathcal{P}_\text{prov}(i = \text{Tom Hanks})_5 & = \text{Patrick Stewart ({\color[HTML]{990099} man})} \\
    \mathcal{P}_\text{prov}(i = \text{Tom Hanks})_6 & = \text{Tom Cruise ({\color[HTML]{990099} man})}.
\end{align} We want to build fair recommendations for the Tom Hanks page, but it is impossible to create a fair recommendation list that contains {\color[HTML]{F69900} woman} by just preprocessing this list. We need to retrieve relevant {\color[HTML]{F69900} woman} accounts from other sources than $\mathcal{P}_\text{prov}(i = \text{Tom Hanks})$.

To overcome this challenge, we use the recommendation network.

\subsection{Recommendation Network}

A recommendation network is a graph where nodes represent items and edges represent recommendation relations. Recommendation networks have been traditionally utilized to investigate the properties of recommender systems \cite{cano2006topology, celma2008new, seyerlehner2009limitation}. They were also used to construct user-side recommender systems \cite{sato2022private}. Recommendation network $G = (V, E)$ we use in this study is defined as follows:
\begin{itemize}
    \item Node set $V$ is the item set $\mathcal{I}$.
    \item Edge set $E$ is defined by the recommendation results of the provider's recommender system. There exists a directed edge from $i \in V$ to $j \in V$ if item $j$ is included in the recommendation list in item $i$, i.e., $\exists k \in [K] \text{ s.t. } \mathcal{P}_\text{prov}(i)_k = j$.
    \item We do not consider edge weights.
\end{itemize}
It should be noted that $G$ can be constructed solely by accessing $\mathcal{P}_\text{prov}$. In other words, an end user can observe $G$.

Recommendation networks have been used in user-side recommender systems \cite{sato2022private, sato2022towards}. However, there is a clear distinction on how to build user-side recommender systems between \textsc{RecCycle} and the previous methods. Previously, the recommendation network was constructed by issuing queries to the official recommender system API or crawling the service site, which incurs a high communication cost. \textsc{RecCycle} does not issue additional queries to the official recommender system. Rather, \textsc{RecCycle} monitors the recommendations presented by the official recommender system when we use the service as usual and partially builds the recommendation network in an online manner. The regions of the network the user has never visited are not available in this manner, but it is not a problem as we will see in the following. \textsc{RecCycle} utilizes this recommendation network to build fair recommendations.

\subsection{RecCycle}

\textsc{RecCycle} can be combined with standard user-side recommender algorithms and turn them into overhead-free. The central idea is to cache the past recommendations.

Suppose we are using the service as usual and visit item page $i \in \mathcal{I}$. The official recommender system shows $\mathcal{P}_\text{prov}(i) = (r_{i1}, r_{i2}, \dots, r_{iK})$ and we can observe this list $(r_{i1}, r_{i2}, \dots, r_{iK})$ by seeing the recommendation slot in the page. \textsc{RecCycle} stores this list into the cache $\tilde{\mathcal{P}}(i) \leftarrow (r_{i1}, r_{i2}, \dots, r_{iK})$. In a practical implementation, \textsc{RecCycle} monitors the Document Object Model (DOM) of the service page by, e.g., a content script of a web browser, and extracts the list of items shown in the recommendation slot. This is a passive process, and it does not issue additional API queries to the official system but just monitors the web page. Some of the official recommendations may not catch the user's eye due to the position or the timing of the presentation, even if the recommendation is attractive to the user. \textsc{RecCycle} caches all of them, including ones the user misses, and utilizes them for future recommendations. The list can be extracted, and the cache $\tilde{\mathcal{P}}$ can be constructed ``for free'', i.e., without any additional communication cost. Note that strictly speaking, $\tilde{\mathcal{P}}$ is not a standard cache because the original queries are not issued by us but by the official system. We extract the results by scraping the DOM and store them as a cache. If we cached the API queries of traditional user-side recommender systems, we would incur API costs at least for the first query. We avoid the communication overhead even for the first recommendation by recycling the results emitted by the official system. When \textsc{RecCycle} creates a recommendation list on the user's side afterward, \textsc{RecCycle} uses $\tilde{\mathcal{P}}(i)$ instead of issuing queries to the official recommender system. This does not incur any communication overhead as $\tilde{\mathcal{P}}(i)$ is stored locally.

The core idea of \textsc{RecCycle} is simple. \textsc{RecCycle} uses any user-side recommender algorithm $\mathcal{A}$ that uses recommendation networks as a backbone algorithm. Many fair user-side recommender algorithms utilizing recommendation networks have been proposed, such as \textsc{PrivateRank} \cite{sato2022private} and \textsc{Consul} \cite{sato2022towards}. We can use any of them as $\mathcal{A}$. When $\mathcal{A}$ issues a query to the official recommender system $\mathcal{P}_{\text{prov}}$, \textsc{RecCycle} uses $\tilde{\mathcal{P}}(i)$ instead. By doing so, \textsc{RecCycle} makes $\mathcal{A}$ free from the API function $\mathcal{P}_{\text{prov}}$ but relies only on the cache $\tilde{\mathcal{P}}$, i.e., overhead-free. One thing we need to be careful about is that $\tilde{\mathcal{P}}(i)$ can be empty for some $i$ if the user has never visited item $i$. We need to implement a fallback mechanism to cope with the empty cache. In many cases, the fallback mechanism can be realized by backtracking the search in the recommendation network.

\setlength{\textfloatsep}{5pt}
\begin{algorithm2e}[t]
    \caption{\textsc{RecCycle} (combined with \textsc{Consul})}
    \label{algo: RecCycle}
    \DontPrintSemicolon
    \nl\KwData{Cached API $\tilde{\mathcal{P}}$, Source item $i \in \mathcal{I}$, Protected attributes $a_i ~\forall i \in \mathcal{I}$, Minimum requirement $\tau$, Set $\mathcal{H}$ of items that user $i$ has already interacted with, Maximum length $L_{\text{max}}$ of search.}
    \nl\KwResult{Recommended items $\mathcal{R} = \{j_k\}_{1 \le k \le K}$.}
    \nl Initialize $\mathcal{R} \leftarrow []$ (empty), $p \leftarrow i$ \;
    \nl $c[a] \leftarrow 0 \quad \forall a \in \mathcal{A}$ \tcp*{counter of sensitive groups}
    \nl $\mathcal{S} \leftarrow \text{Stack}([])$ \tcp*{empty stack}
    \nl \For{\textup{iter} $\gets 1$ \textbf{to} $L_{\text{max}}$}{
        \nl     \While{$p$ \textup{has already been visited or} $\tilde{\mathcal{P}}(p)$ \textup{is not cached}}{

            \nl     \If{$|\mathcal{S}| = 0$}{
                \nl         \textbf{goto} line 21 \tcp*{no further items}
            }
            \nl     $p \leftarrow \mathcal{S} \text{.pop\_top}()$  \tcp*{next search node}
        }
        \nl     \For{$k \gets 1$ \textbf{to} $K$}{
            \nl         $j \leftarrow \tilde{\mathcal{P}}(p)_{k}$ \;
            \nl         \If{$j \textup{ \textbf{not} \textbf{in} } \mathcal{R} \cup \mathcal{H}$ \textup{\textbf{and}} $\sum_{a \neq a_j} \max(0, \tau - c[a]) \le K - |\mathcal{R}| - 1$}{
                \nl             \tcc{$j$ can be safely added keeping fairness. Avoid items in $\mathcal{R} \cup \mathcal{H}$.}
                \nl             Push back $j$ to $\mathcal{R}$. \;
                \nl             $c[a_j] \leftarrow c[a_j] + 1$ \;
            }
            \nl         \If{$|\mathcal{R}| = K$}{
                \nl             \textbf{return} $\mathcal{R}$ \tcp*{list is full}
            }
        }
        \nl     \For{$k \gets K$ \textbf{to} $1$}{
            \nl         $\mathcal{S} \text{.push\_top}(\tilde{\mathcal{P}}(p)_{k})$ \tcp*{insert candidates}
        }
    }
    \nl \While{$|\mathcal{R}| < K$}{
        \nl     $j \leftarrow$ \text{Uniform}($\mathcal{I}$) \tcp*{random item}
        \nl     \If{$j \textup{ \textbf{not} \textbf{in} } \mathcal{R} \cup \mathcal{H}$ \textup{\textbf{and}} $\sum_{a \neq a_j} \max(0, \tau - c[a]) \le K - |\mathcal{R}| - 1$}{
            \nl         Push back $j$ to $\mathcal{R}$. \;
            \nl         $c[a_j] \leftarrow c[a_j] + 1$ \;
        }
    }
    \nl \textbf{return} $\mathcal{R}$ \;
\end{algorithm2e}

Algorithm \ref{algo: RecCycle} shows the pseudocode of \textsc{RecCycle} combined with \textsc{Consul} \cite{sato2022towards}. We use $\tau \in \mathbb{Z}_{\ge 0}$ to control the trade-off between fairness and performance. $\tau$ is the minimum number of items for each sensitive group. $\tau = 0$ indicates that the recommender system does not care about fairness because no constraint is imposed. An increase in $\tau$ should lead to a corresponding increase in fairness.

The algorithm is a depth-first search algorithm that searches the recommendation network. When the hit item is not cached, it backtracks the search. In lines 21--25, the fallback process ensures that $K$ items are recommended following \textsc{Consul} \cite{sato2022towards} to ensure the soundness of the algorithm.

\begin{figure*}[t]
    \includegraphics[width=0.7\linewidth]{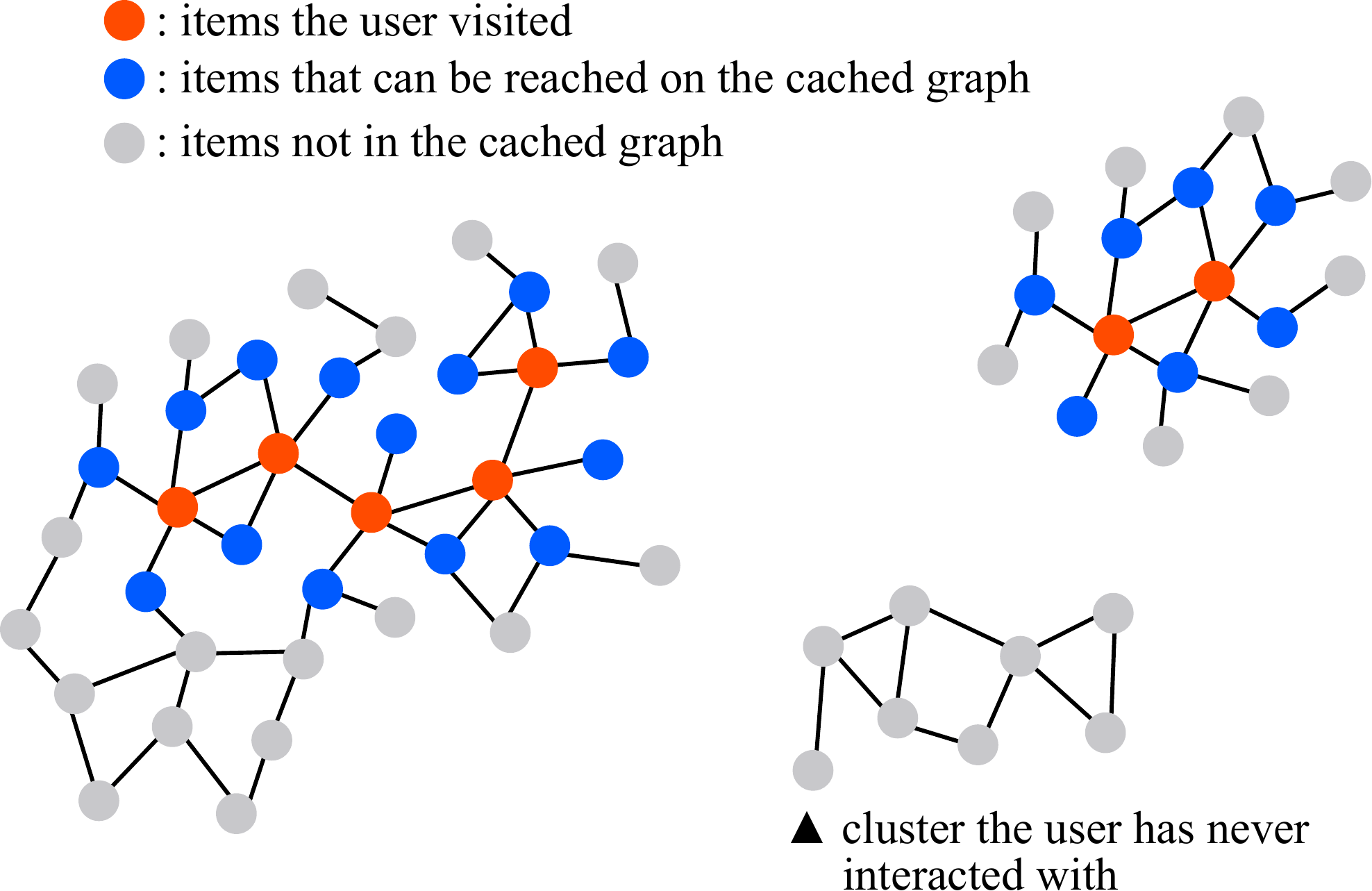}
    \caption{Illustration of the cached recommendation network. When the user visits a page represented by a node, surrounding nodes are recommended and observed in the DOM. The boundary items of the visited items can be reached by the search in the cached recommendation network. These items are likely to be relevant to the user. There are clusters that the user has never visited, and the cached graph does not contain them, but this is not a problem for effective recommendations in most cases.}
    \label{fig: cached_graph}
\end{figure*}

\textbf{Discussion on the size of the cache.} The number of the cached items stored in cache $\tilde{\mathcal{P}}$ is far less than the number of all items in the service. For example, Amazon contains more than $600$ million products, but ordinary users interact with only a few hundred products. \textsc{RecCycle} can use such a small number of items only. It may seem that a small cache size may lead to poor recommendations, but such a small number of items is sufficient in practice. First, most of the items in the service are low-quality and/or irrelevant to the user. It is not likely that we need to take all items into account. Second, users usually walk around the regions they are interested in. The gathered items in the cache are likely to be relevant to the user. Third, the gathered lists $\tilde{\mathcal{P}}(i)$ contain not only the items the user has already visited but also the items that are one-hop away from the interacted items. Therefore, the boundary of the items the user has visited can be reached by the search in the cached recommendation network, and \textsc{RecCycle} can recommend them (Figure \ref{fig: cached_graph}). The boundary items (shown in blue in Figure \ref{fig: cached_graph}) are likely to be relevant to the user, and the other items (shown in gray) are not likely to be necessary for effective recommendations. 

\subsection{Properties of RecCycle} \label{sec: properties}

The properties of user-side recommender systems are defined as follows:

\vspace{0.1in}
\noindent \textbf{Consistency.} A user-side recommender system $\mathcal{Q}$ is consistent if nDCG of $\mathcal{Q}$ with $\tau = 0$ is guaranteed to be the same as that of the official recommender system. In other words, a consistent user-side recommender system does not degrade the performance if we do not impose the fairness constraint.

\vspace{0.1in}
\noindent \textbf{Soundness.} We say a user-side recommender system is sound if the minimum number of items from each sensitive group is guaranteed to be at least $\tau$ provided $0 \le \tau \le K/|\mathcal{A}|$ and there exist at least $\tau$ items for each group $a \in \mathcal{A}$\footnote{Note that if $\tau > K/|\mathcal{A}|$ or there exist less than $\tau$ items for some group $a$, it is impossible to achieve this condition.}. In other words, we can provably guarantee the fairness of a sound user-side recommender system by adjusting $\tau$.

\vspace{0.1in}
\noindent \textbf{Locality.} We say a user-side recommender system is local if it generates recommendations without loading the entire recommendation network. Locality is important for both computational and communication efficiency.

\vspace{0.1in}
\noindent \textbf{Overhead-free.} We say a user-side recommender system is overhead-free if it does not access the official recommender system $\mathcal{P}_{\text{prov}}$ when making a recommendation list. Overhead-freeness is important for communication efficiency.

\begin{theorem} \label{thm: recrecycly}
    \textsc{RecCycle} with \textsc{Consul} is consistent, sound, local, and overhead-free.
\end{theorem}

\begin{proof}
    The proof is similar to that of \cite{sato2022towards}. We repeat the proof for completeness.

    \noindent \textbf{Consistency.} If $\tau = 0$, $\sum_{a \neq a_j} \max(0, \tau - c[a]) = 0$ holds in line 13. Therefore, the condition in line 13 passes in all $K$ iterations in the initial node, and the condition in line 17 passes at the $K$-th iteration. Besides, $\tilde{\mathcal{P}}(i)$ is guaranteed to be cached as this is the exact item that the user is currently viewing. The final output is $\mathcal{R} = \tilde{\mathcal{P}}(i) = \mathcal{P}_{\text{prov}}(i)$.

    \vspace{0.1in}
    \noindent \textbf{Soundness.} We prove by mathematical induction that \begin{align}\sum_{a \in \mathcal{A}} \max(0, \tau - c[a]) \le K - |\mathcal{R}| \label{eq: consul-proof-eq}\end{align} holds in every step in Algorithm \ref{algo: RecCycle}. At the initial state, $|\mathcal{R}| = 0$, $c[a] = 0 ~\forall a \in \mathcal{A}$, and $\sum_{a \in \mathcal{A}} \max(0, \tau - c[a]) = \tau |\mathcal{A}|$. Thus, the inequality holds by the assumption $\tau |\mathcal{A}| \le K$. The only steps that alter the inequality are lines 15--16 and lines 24-25. Let $c, \mathcal{R}$ be the states before the execution of these steps, and $c', \mathcal{R}'$ be the states after the execution. We prove that \begin{align}\sum_{a \in \mathcal{A}} \max(0, \tau - c'[a]) \le K - |\mathcal{R}'|\end{align} holds assuming \begin{align}\sum_{a \in \mathcal{A}} \max(0, \tau - c[a]) \le K - |\mathcal{R}|, \label{eq: consul-proof-inductive}\end{align} i.e., the inductive hypothesis.  When these steps are executed, the condition \begin{align}\sum_{a \neq a_j} \max(0, \tau - c[a]) \le K - |\mathcal{R}| - 1 \label{eq: consul-proof-if}\end{align} holds by the conditions in lines 13 and 23. We consider two cases. (i) If $c[a_j] \ge \tau$ holds,
    \begin{align*}
    \sum_{a \in \mathcal{A}} \max(0, \tau - c'[a])
    &\stackrel{\text{(a)}}{=} \sum_{a \neq a_j} \max(0, \tau - c'[a]) \\
    &\stackrel{\text{(b)}}{=} \sum_{a \neq a_j} \max(0, \tau - c[a]) \\
    &\stackrel{\text{(c)}}{\le} K - |\mathcal{R}| - 1 \\
    &= K - |\mathcal{R}'|,
    \end{align*}
    where (a) follows $c[a_j] \ge \tau$, (b) follows $c'[a] = c[a] ~(\forall a \neq a_j)$, and (c) follows eq. \eqref{eq: consul-proof-if}. (ii) If $c[a_j] < \tau$ holds,
    \begin{align*}
    \sum_{a \in \mathcal{A}} \max(0, \tau - c'[a])
    &\stackrel{\text{(a)}}{=} \sum_{a \in \mathcal{A}} \max(0, \tau - c[a]) - 1 \\
    &\stackrel{\text{(b)}}{\le} K - |\mathcal{R}| - 1 \\
    &= K - |\mathcal{R}'|,
    \end{align*}
    where (a) follows $c'[a_j] = c[a_j] + 1$ and $c[a_j] + 1 \le \tau$, and (b) follows eq. \eqref{eq: consul-proof-inductive}. In sum, eq. \eqref{eq: consul-proof-eq} holds by mathematical induction. When Algorithm \ref{algo: RecCycle} terminates, $|\mathcal{R}| = K$. As the left hand side of eq. \eqref{eq: consul-proof-eq} is non-negative, each term should be zero. Thus, $c[a] = |\{i \in \mathcal{R} \mid a_i = a\}| \ge \tau$ holds for all $a \in \mathcal{A}$.
    
    \vspace{0.1in}
    \noindent \textbf{Locality.} \textsc{RecCycle} accesses the recommendation network in lines 12 and 20. As the query item $p$ changes at most $L_\text{max}$ times in line 10, \textsc{RecCycle} accesses at most $L_\text{max}$ items, which is a constant, among $n$ items.

    \vspace{0.1in}
    \noindent \textbf{Overhead-free.} \textsc{RecCycle} is overhead-free because it does not access the official recommender system $\mathcal{P}_{\text{prov}}$ in Algorithm \ref{algo: RecCycle}.
\end{proof}

\section{Experiments} \label{sec: experiments}

We answer the following questions through the experiments.

\begin{itemize}
    \item (RQ1) How good a trade-off between performance and efficiency does \textsc{RecCycle} strike?
    \item (RQ2) Is \textsc{RecCycle} robust to sparse cache?
    \item (RQ3) Does \textsc{RecCycle} work in the real world?
\end{itemize}

\begin{figure*}[t]
    \captionsetup{labelformat=empty}
    \caption{Table 3. Performance Comparison. Cost denotes the average number of times each method accesses item pages, i.e., the number of queries to the official recommender systems. The less this value is, the more communication-efficient the method is. The best score is highlighted in bold among the four methods (excluding the Oracle method). \textsc{RecCycle} is extremely more efficient than other methods by achieving zero communication overhead (i.e., it accesses only one page, the item page the user is currently viewing) while it achieves on par or slightly worse performances than the Oracle method and state-of-the-art user-side recommender systems.}

    \includegraphics[width=\linewidth]{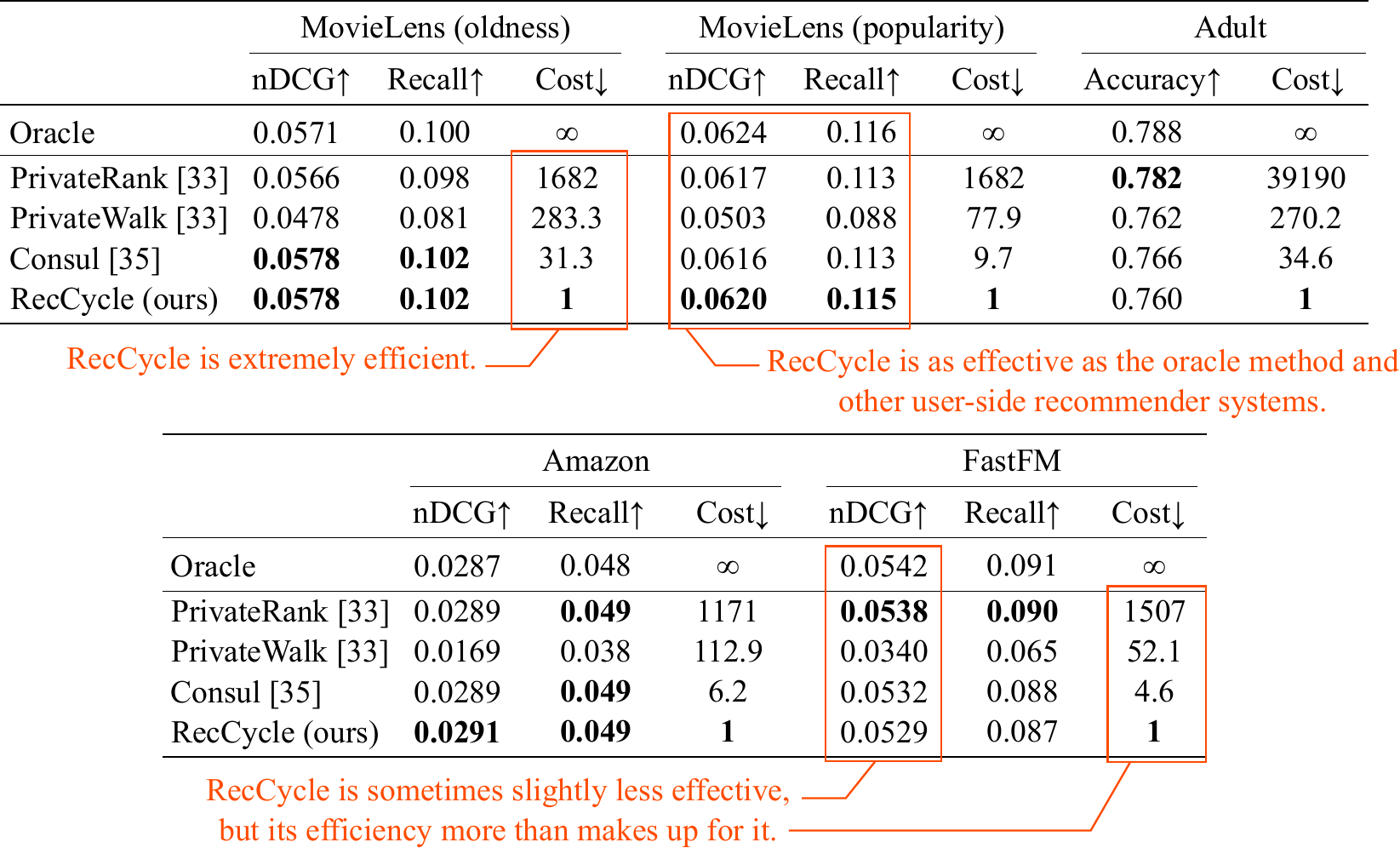}
    \label{tab: performance}
\end{figure*}
\addtocounter{table}{1}
\addtocounter{figure}{-1}

\subsection{(RQ1) Performance} \label{sec: experiments-performance}

\noindent \textbf{Setup.} We use MovieLens100k \cite{harper2016movielens}, Adult dataset, Amazon Home and Kitchen \cite{he2016ups, mcauley2015image}, and LastFM \footnote{\url{https://grouplens.org/datasets/hetrec-2011/}} datasets following the existing work \cite{sato2022private, sato2022towards}.

\noindent \textbf{MovieLens dataset.} In this dataset, an item represents a movie. We consider two ways of creating protected groups, (i) oldness: We regard movies released before 1990 as a protected group, and (ii) popularity: We regard movies with less than $50$ reviews as the protected group. We use Bayesian personalized ranking (BPR) \cite{rendle2009bpr} for the provider's recommender system, where the similarity of items is defined by the inner product of the latent vectors of the items, and the top-$K$ similar items are recommended. We use the default parameters of Implicit package\footnote{\url{https://github.com/benfred/implicit}} for BPR. We measure nDCG@$K$ and recall@$K$ as performance metrics following previous works \cite{krichene2020sampled, he2017neural, rendle2009bpr, sato2022private}. Note that we use the full datasets to compute nDCG and recall instead of employing negative samples to avoid biased evaluations \cite{krichene2020sampled}.

\noindent \textbf{Adult dataset.} In this dataset, an item represents a person, and the sensitive attribute is defined by sex. We use the nearest neighbor recommendations with demographic features, including age, education, and capital-gain, as the provider's official recommender system. The label of an item (person) represents whether the income exceeds \$50\,000 per year. The accuracy for item $i$ represents the ratio of the recommended items for item $i$ that have the same label as item $i$. The overall accuracy is the average of the accuracy of all items.

\noindent \textbf{LastFM and Amazon dataset.} In these datasets, an item represents a music and a product, respectively. We regard items that received less than $50$ interactions as a protected group. We extract $10$-cores for these datasets by iteratively discarding items and users with less that $10$ interactions. We use BPR for the provider's official recommender system, as on the MovieLens dataset. We use nDCG@$K$ and recall@$K$ as performance metrics.

In all datasets, we set $K = 10$ and $\tau = 5$, i.e., recommend $5$ protected items and $5$ other items. Note that all methods, \textsc{RecCycle} and the baselines, are guaranteed to generate completely balanced recommendations, i.e., they recommend $5$ protected items and $5$ other items, when we set $K = 10$ and $\tau = 5$. We checked that the results were completely balanced. So we do not report the fairness scores but only report performance metrics in this experiment.

\textsc{RecCycle} requires the cache. As these datasets do not contain user history data, we simulate the user behavior by a random walk on the recommendation network, assuming users surf the service site by following recommendations, and create the cache by storing the past recommendations for the items the user saw. This simulation is reasonable because users often find items by recommender systems, e.g., $75$ \% of watches in Netflix came from recommender systems \cite{mackenzie2013how}. We set the length of the random walk (i.e., the number of items in the user history) to $100$ for all datasets, which is a reasonable number of items that users interact with in practice.

\vspace{0.1in}
\noindent \textbf{Methods.} We use \textsc{PrivateRank} and \textsc{PrivateWalk} \cite{sato2022private} and \text{Consul} \cite{sato2022towards} as baseline methods. We use the default hyperparameters reported in the original papers for these methods. Note that \textsc{PrivateRank} and \textsc{PrivateWalk} were reported to be insensitive to the choice of these hyperparameters \cite{sato2022private}. In addition, we use the oracle method \cite{sato2022private} that uses the complete similarity scores used in the provider's recommender system and adopts the same fair postprocessing as \textsc{PrivateRank}. The oracle method uses hidden information, which is not observable by end users, and can be seen as an ideal upper bound of performance.

\vspace{0.1in}
\noindent \textbf{Results.} Table 3 reports performances and efficiency of each method\footnote{The values are inconsistent with ones reporeted in \cite{sato2022towards}. We found a bug in the evaluation script used in \cite{sato2022towards}. It divided the sum of recalls and nDCGs by the number $n$ of items to compute the mean, but it should be divided by the number $m$ of users. We fixed it in this table. It only affects the scale of values, and the tendency is the same.}. First, \textsc{RecCycle} is the most efficient method and requires only one query to generate a recommendation list while the most efficient baseline, \textsc{Consul} requires $4$ to $6$ times more queries to the official recommender system. The only one page \textsc{RecCycle} requires is the item page the user is currently viewing, which is necessary anyway, even without \textsc{RecCycle}, and it means that \textsc{RecCycle} achieves zero communication overhead. Note that \textsc{PrivateWalk} and \textsc{Consul} were initially proposed as an efficient method. Nevertheless, \textsc{RecCycle} further improves communication costs by a large margin. Second, \textsc{RecCycle} performs on par or slightly worse than Oracle, \textsc{Consul}, and \textsc{PrivateRank}, whereas \textsc{PrivateWalk} degrades performance in exchange for its efficiency. In sum, \textsc{RecCycle} strikes an excellent trade-off between performance and communication costs.

\begin{figure}
    \centering
    \includegraphics[width=0.8\linewidth]{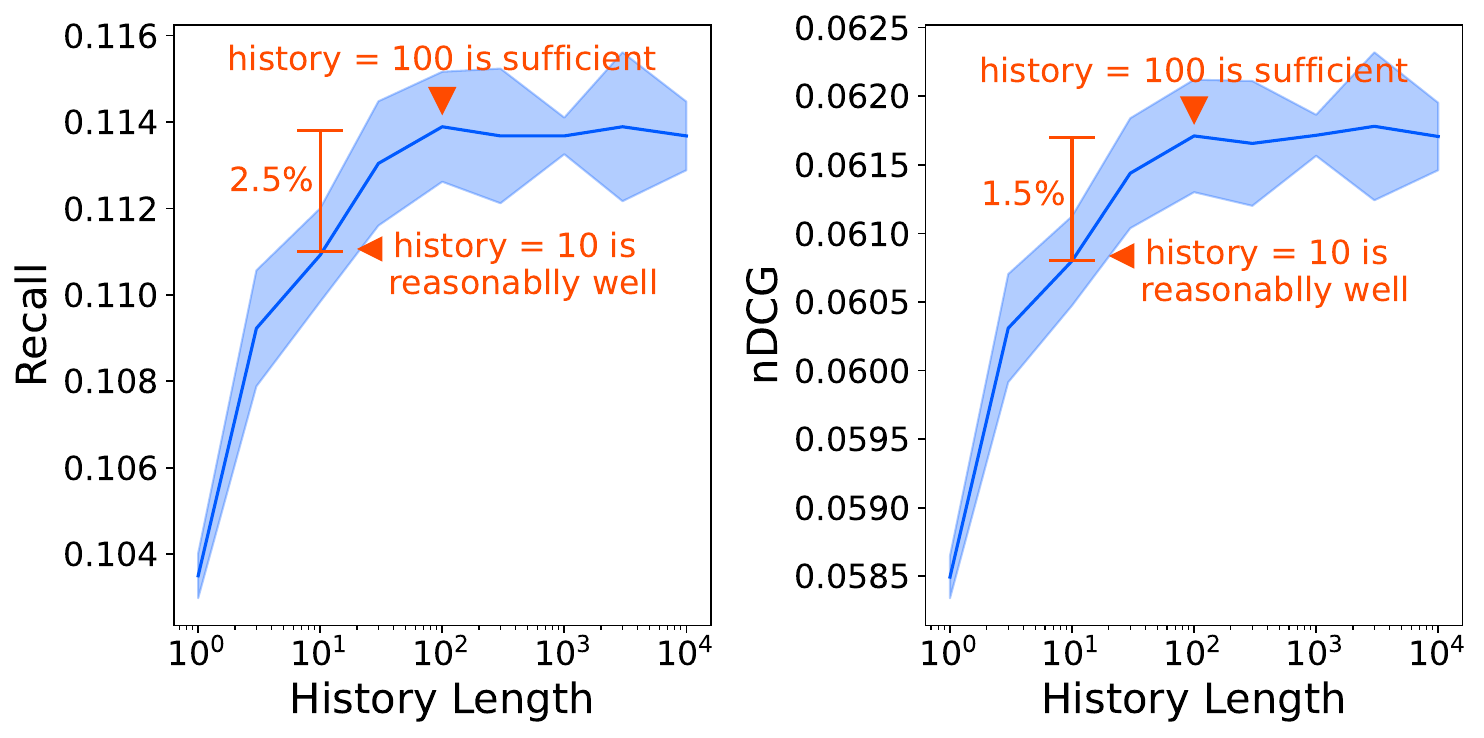}
    \caption{Performance of \textsc{RecCycle} with different lengths of user histories. Note that the y-axis starts at around $0.102$ for recall and $0.058$ for nDCG, and the relative difference in performance is within a few percent. \textsc{RecCycle} achieves good performance even when the user history is short and the cache is sparse.}
    \label{fig: sparsity}
\end{figure}

\subsection{(RQ2) Robustness against Sparsity} \label{sec: sparsity}

We investigate the robustness of \textsc{RecCycle} against sparse history. \textsc{RecCycle} relies on the cache stored in the user's device. The cache may be sparse when the user has interacted with few items, i.e., the cold start problem in the user-side recommender systems. We evaluate the performance of \textsc{RecCycle} with different lengths of user histories.

The problem setup is the same as in Section \ref{sec: experiments-performance}, but we vary the lengths of random walks, i.e., the length of user histories. A shorter random walk indicates the user is newer or less active and has interacted with fewer items. We show the results for MovieLens + popularity fairness as the tendency is the same for all datasets. Figure \ref{fig: sparsity} shows the performance of \textsc{RecCycle} with different lengths of user histories. We observe that the performance of \textsc{RecCycle} is robust against the sparsity of interaction. Note that the y-axis starts at around $0.102$ for recall and $0.058$ for nDCG, and the relative difference in performance is within a few percent. The performance is slightly worse when the cache is sparse, but the drop is marginal. Overall, $100$ items are sufficient to achieve a good performance, which is a realistic number of items that users interact with in practice, and a few to tens of items are sufficient to achieve a reasonable performance, which means that \textsc{RecCycle} is robust against the cold start problem.

\subsection{(RQ3) Case Study in the Real World} \label{sec: twitter}

\begin{table}[t]
    \caption{Recommendations for ``Tom Hanks'' on the real-world X (Twitter) environment. The provider's official system recommends only man accounts and is not fair with respect to gender. All of \textsc{PrivateWalk}, \textsc{Consul}, and \textsc{RecCycle} succeed in completely balanced recommendations. All the recommended users by \textsc{RecCycle} are actors/actresses and related to the source user, Tom Hanks. Moreover, \textsc{RecCycle} requires only the information of the one page that the user is currently viewing, and is more efficient than \textsc{PrivateWalk} and \textsc{Consul}, which require several additional queries to the official system of X.}
    \centering
    \vspace{-0.15in}
    \scalebox{0.72}{
        \begin{tabular}{lcccc} \toprule
            Source Item & \multicolumn{4}{c}{Tom Hanks}                                                                                                         \\ \midrule
            Method & Provider                                      & \textsc{PrivateWalk}                               & \textsc{Consul}                                    & \textsc{RecCycle}     \\
            Cost & 1 access                                          & 17 accesses                                        & 4 accesses                                         & 1 access (overhead-free)    \\ \midrule
            & Seth Macfarlane ({\color[HTML]{990099} man}) & Seth Macfarlane ({\color[HTML]{990099} man})      & Seth Macfarlane ({\color[HTML]{990099} man})      & Seth Macfarlane ({\color[HTML]{990099} man}) \\
            & Danny Devito ({\color[HTML]{990099} man})     & Seth Meyers ({\color[HTML]{990099} man})           & Danny Devito ({\color[HTML]{990099} man})          & Danny Devito ({\color[HTML]{990099} man}) \\
            & Leonardo Dicaprio ({\color[HTML]{990099} man})& Tom Cruise ({\color[HTML]{990099} man})            & Leonardo Dicaprio ({\color[HTML]{990099} man})     & Leonardo Dicaprio ({\color[HTML]{990099} man}) \\
            & Jason Bateman ({\color[HTML]{990099} man})    & Alia Bhatt ({\color[HTML]{F69900} woman})          & Julia Louis-Dreyfus ({\color[HTML]{F69900} woman}) &  Rihanna ({\color[HTML]{F69900} woman}) \\
            & Patrick Stewart ({\color[HTML]{990099} man})  & Christina Applegate ({\color[HTML]{F69900} woman}) & Christina Applegate ({\color[HTML]{F69900} woman}) & Lady Gaga ({\color[HTML]{F69900} woman})         \\
            & Tom Cruise ({\color[HTML]{990099} man})       & Candace Owens ({\color[HTML]{F69900} woman})        & Anna Kendrick ({\color[HTML]{F69900} woman})       & Katy Perry ({\color[HTML]{F69900} woman})  \\ \bottomrule
        \end{tabular}
    }
    \label{tab: twitter}
\end{table}

Finally, we show that \textsc{RecCycle} is applicable to real-world services via a case study. We use the user recommender system on X (Twitter), which is in operation in the real world. We stress that we are not employees of X Corp. and do not have any access to the hidden data stored in X. We run the experiments as an ordinary end user. The sensitive attribute is defined by gender in this experiment. Table \ref{tab: twitter} shows the results for the account of ``Tom Hanks.'' The official recommender system shows only man accounts and is not fair with respect to gender. By contrast, \textsc{RecCycle}'s recommendations are completely balanced. In addition, \textsc{RecCycle} is overhead-free, i.e., it requires only the information of the one page that the user is currently viewing, and is more efficient than \textsc{PrivateWalk} and \textsc{Consul}, which require several additional queries to the official system of X. The results show that we can build a fair recommender system for X with respect to gender and realize the functionality we call for even though we are not employees but ordinary end users. 

\section{Related Work}

\noindent \textbf{User-side Realization.} User-side realization \cite{sato2022private, sato2024user} aims to help end users of services realize the functionalities they want, whereas traditional methods are designed for the developers of the services. Since the service is not tailor-made for a user, it is natural for dissatisfaction to arise. User-side realization provides proactive solutions to this problem. There have been many algorithms for user-side realization, including user-side search engines \cite{sato2022retrieving,sato2022clear}, user-side privacy protection \cite{sato2024making}, user-side watermarking of large language models \cite{sato2023embarrassingly}.

Among them, user-side recommender systems \cite{sato2022private, sato2022towards} are the most relevant to this study. \textsc{PrivateWalk} and \textsc{PrivateRank} are the first user-side recommender systems \cite{sato2022private}. However, they are not practical because they require a large number of accesses to the official recommender system or degrade the performance. \textsc{Consul} is the state-of-the-art user-side recommender system that strikes an excellent trade-off between performance and efficiency. \textsc{Consul} is also theoretically grounded and is the first user-side recommendation algorithm that is proven to hold three desired properties: consistency, soundness, and locality. Although \textsc{Consul} is the best user-side recommender system in the literature, it still incurs non-negligible communication overhead and is not practical in the real world.

Our proposed method, \textsc{RecCycle}, is the first overhead-free user-side recommender system. Our method is built on \textsc{Consul} and inherits the properties while drastically improving the communication efficiency.

\vspace{0.1in}
\noindent \textbf{Fairness in Recommender Systems.} As fairness has become a major concern in society \cite{united2014big, executive2016big}, many fairness-aware machine learning algorithms have been proposed \cite{hardt2016equality, kamishima2012fairness, zafar2017fairness}. In particular, fairness with respect to gender \cite{zehlike2017fair, singh2018fairness, xu2020algorithmic}, race \cite{zehlike2017fair, xu2020algorithmic}, financial status \cite{fu2020fairness}, and popularity \cite{mehrotra2018towards, xiao2019beyond} is of great concern. In light of this, many fairness-aware recommendation algorithms have been proposed \cite{kamishima2012enhancement, yao2017beyond, biega2018equity, milano2020recommender}. Some of them aim to ensure fairness for users \cite{bruke2017mltisided} and others aim to ensure fairness for items (such as products and accounts) \cite{ekstrand2018exploring, beutel2019fairness, mehrotra2018towards, liu2019personalized, geyik2018building}, and some aim to ensure fairness for both users and items \cite{bruke2017mltisided}. In this study, we focus on fairness for items following \cite{sato2022private, sato2022towards}, but the proposed method can also be applied to fairness for users when the official system is not fair for the very user that uses the user-side system. The user can overcome the unfairness of the official system by creating their own recommender system. Note that fairness is closely related to topic diversification \cite{ziegler2005improving} by regarding the topic as the sensitive attribute, and we considered the diversity of recommended items in this study as well.

\vspace{0.1in}
\noindent \textbf{Steerable Recommender Systems.} The reliability of recommender systems has attracted a lot of attention \cite{tintarev2007survey, balog2019transparent}, and steerable recommender systems that let the users modify the behavior of the system have been proposed \cite{green2009generating, balog2019transparent}. User-side recommender systems also allow the users to modify the recommendation results. However, the crucial difference between steerable and user-side recommender systems is that steerable recommender systems must be implemented by a service provider, whereas user-side recommender systems can be built by arbitrary users even if the official system is an ordinary (non-steerable) one. Therefore, user-side recommender systems can expand the scope of steerable recommender systems by a considerable margin \cite{sato2022private}.

\section{Conclusion}

We proposed \textsc{RecCycle}, the first overhead-free user-side recommender system. \textsc{RecCycle} realizes the functionalities a user wants on the user's side without communication overhead to the official system, which is in stark contrast to the existing user-side recommender systems, which require additional queries to the official system. This property is beneficial both for the users and the service providers. The users can enjoy the user-side system in real-time as if it were the official system, and the service providers can reduce the load of the servers. In the experiments, we confirmed that \textsc{RecCycle} is as effective as existing user-side recommender systems while it is much more efficient. We also carried out a case study on the real-world Twitter environment and showed that \textsc{RecCycle} can be applied to real-world services. We believe that \textsc{RecCycle} will be a practical and useful tool for end users to build their own recommender systems.

\bibliographystyle{plainnat}
\bibliography{sample-base}

\begin{thebibliography}{50}
\providecommand{\natexlab}[1]{#1}
\providecommand{\url}[1]{\texttt{#1}}
\expandafter\ifx\csname urlstyle\endcsname\relax
  \providecommand{\doi}[1]{doi: #1}\else
  \providecommand{\doi}{doi: \begingroup \urlstyle{rm}\Url}\fi

\bibitem[Anderson et~al.(2020)Anderson, Maystre, Anderson, Mehrotra, and Lalmas]{anderson2020algorithmic}
Ashton Anderson, Lucas Maystre, Ian Anderson, Rishabh Mehrotra, and Mounia Lalmas.
\newblock Algorithmic effects on the diversity of consumption on spotify.
\newblock In \emph{{WWW}}, pages 2155--2165, 2020.

\bibitem[Balog et~al.(2019)Balog, Radlinski, and Arakelyan]{balog2019transparent}
Krisztian Balog, Filip Radlinski, and Shushan Arakelyan.
\newblock Transparent, scrutable and explainable user models for personalized recommendation.
\newblock In \emph{{SIGIR}}, pages 265--274, 2019.

\bibitem[Beutel et~al.(2019)Beutel, Chen, Doshi, Qian, Wei, Wu, Heldt, Zhao, Hong, Chi, and Goodrow]{beutel2019fairness}
Alex Beutel, Jilin Chen, Tulsee Doshi, Hai Qian, Li~Wei, Yi~Wu, Lukasz Heldt, Zhe Zhao, Lichan Hong, Ed~H. Chi, and Cristos Goodrow.
\newblock Fairness in recommendation ranking through pairwise comparisons.
\newblock In \emph{{KDD}}, pages 2212--2220, 2019.

\bibitem[Biega et~al.(2018)Biega, Gummadi, and Weikum]{biega2018equity}
Asia~J. Biega, Krishna~P. Gummadi, and Gerhard Weikum.
\newblock Equity of attention: Amortizing individual fairness in rankings.
\newblock In \emph{{SIGIR}}, pages 405--414, 2018.

\bibitem[Burke(2017)]{bruke2017mltisided}
Robin Burke.
\newblock Multisided fairness for recommendation.
\newblock In \emph{4th Workshop on Fairness, Accountability, and Transparency in Machine Learning, {FAT/ML}}, 2017.

\bibitem[Cano et~al.(2006)Cano, Celma, Koppenberger, and Buldu]{cano2006topology}
Pedro Cano, Oscar Celma, Markus Koppenberger, and Javier~M Buldu.
\newblock Topology of music recommendation networks.
\newblock \emph{Chaos: An interdisciplinary journal of nonlinear science}, 16\penalty0 (1):\penalty0 013107, 2006.

\bibitem[Celma and Herrera(2008)]{celma2008new}
{\`{O}}scar Celma and Perfecto Herrera.
\newblock A new approach to evaluating novel recommendations.
\newblock In \emph{RecSys}, pages 179--186, 2008.

\bibitem[Chen et~al.(2021)Chen, Wang, Xu, Le, Sharma, Richardson, Wu, and Chi]{chen2021values}
Minmin Chen, Yuyan Wang, Can Xu, Ya~Le, Mohit Sharma, Lee Richardson, Su{-}Lin Wu, and Ed~H. Chi.
\newblock Values of user exploration in recommender systems.
\newblock In \emph{Proceedings of the 15th {ACM} Conference on Recommender Systems, {RecSys}}, pages 85--95, 2021.

\bibitem[Ekstrand et~al.(2018)Ekstrand, Tian, Kazi, Mehrpouyan, and Kluver]{ekstrand2018exploring}
Michael~D. Ekstrand, Mucun Tian, Mohammed R.~Imran Kazi, Hoda Mehrpouyan, and Daniel Kluver.
\newblock Exploring author gender in book rating and recommendation.
\newblock In \emph{RecSys}, pages 242--250, 2018.

\bibitem[Fu et~al.(2020)Fu, Xian, Gao, Zhao, Huang, Ge, Xu, Geng, Shah, Zhang, and de~Melo]{fu2020fairness}
Zuohui Fu, Yikun Xian, Ruoyuan Gao, Jieyu Zhao, Qiaoying Huang, Yingqiang Ge, Shuyuan Xu, Shijie Geng, Chirag Shah, Yongfeng Zhang, and Gerard de~Melo.
\newblock Fairness-aware explainable recommendation over knowledge graphs.
\newblock In \emph{{SIGIR}}, pages 69--78, 2020.

\bibitem[Geyik and Kenthapadi(2018)]{geyik2018building}
Sahin~Cem Geyik and Krishnaram Kenthapadi.
\newblock Building representative talent search at {LinkedIn}, 2018.
\newblock URL \url{https://engineering.linkedin.com/blog/2018/10/building-representative-talent-search-at-linkedin}.

\bibitem[Geyik et~al.(2019)Geyik, Ambler, and Kenthapadi]{geyik2019fairness}
Sahin~Cem Geyik, Stuart Ambler, and Krishnaram Kenthapadi.
\newblock Fairness-aware ranking in search {\&} recommendation systems with application to linkedin talent search.
\newblock In \emph{{KDD}}, pages 2221--2231, 2019.

\bibitem[Green et~al.(2009)Green, Lamere, Alexander, Maillet, Kirk, Holt, Bourque, and Mak]{green2009generating}
Stephen~J. Green, Paul Lamere, Jeffrey Alexander, Fran{\c{c}}ois Maillet, Susanna Kirk, Jessica Holt, Jackie Bourque, and Xiao{-}Wen Mak.
\newblock Generating transparent, steerable recommendations from textual descriptions of items.
\newblock In \emph{{RecSys}}, pages 281--284, 2009.

\bibitem[Hardt et~al.(2016)Hardt, Price, and Srebro]{hardt2016equality}
Moritz Hardt, Eric Price, and Nati Srebro.
\newblock Equality of opportunity in supervised learning.
\newblock In \emph{{NeurIPS}}, pages 3315--3323, 2016.

\bibitem[Harper and Konstan(2016)]{harper2016movielens}
F.~Maxwell Harper and Joseph~A. Konstan.
\newblock The movielens datasets: History and context.
\newblock \emph{{ACM} Trans. Interact. Intell. Syst.}, 5\penalty0 (4):\penalty0 19:1--19:19, 2016.

\bibitem[He and McAuley(2016)]{he2016ups}
Ruining He and Julian~J. McAuley.
\newblock Ups and downs: Modeling the visual evolution of fashion trends with one-class collaborative filtering.
\newblock In \emph{{WWW}}, pages 507--517, 2016.

\bibitem[He et~al.(2017)He, Liao, Zhang, Nie, Hu, and Chua]{he2017neural}
Xiangnan He, Lizi Liao, Hanwang Zhang, Liqiang Nie, Xia Hu, and Tat{-}Seng Chua.
\newblock Neural collaborative filtering.
\newblock In \emph{{WWW}}, pages 173--182, 2017.

\bibitem[Kamishima et~al.(2012{\natexlab{a}})Kamishima, Akaho, Asoh, and Sakuma]{kamishima2012enhancement}
Toshihiro Kamishima, Shotaro Akaho, Hideki Asoh, and Jun Sakuma.
\newblock Enhancement of the neutrality in recommendation.
\newblock In \emph{Proceedings of the 2nd Workshop on Human Decision Making in Recommender Systems}, volume 893, pages 8--14, 2012{\natexlab{a}}.

\bibitem[Kamishima et~al.(2012{\natexlab{b}})Kamishima, Akaho, Asoh, and Sakuma]{kamishima2012fairness}
Toshihiro Kamishima, Shotaro Akaho, Hideki Asoh, and Jun Sakuma.
\newblock Fairness-aware classifier with prejudice remover regularizer.
\newblock In \emph{{ECML} {PKDD}}, volume 7524, pages 35--50, 2012{\natexlab{b}}.

\bibitem[Krichene and Rendle(2020)]{krichene2020sampled}
Walid Krichene and Steffen Rendle.
\newblock On sampled metrics for item recommendation.
\newblock In \emph{{KDD}}, pages 1748--1757, 2020.

\bibitem[Linden et~al.(2003)Linden, Smith, and York]{linden2003amazon}
Greg Linden, Brent Smith, and Jeremy York.
\newblock Amazon. com recommendations: Item-to-item collaborative filtering.
\newblock \emph{IEEE Internet computing}, 7\penalty0 (1):\penalty0 76--80, 2003.

\bibitem[Liu et~al.(2019)Liu, Guo, Sonboli, Burke, and Zhang]{liu2019personalized}
Weiwen Liu, Jun Guo, Nasim Sonboli, Robin Burke, and Shengyu Zhang.
\newblock Personalized fairness-aware re-ranking for microlending.
\newblock In \emph{RecSys}, pages 467--471, 2019.

\bibitem[MacKenzie et~al.(2013)MacKenzie, Meyer, and Noble]{mackenzie2013how}
Ian MacKenzie, Chris Meyer, and Steve Noble.
\newblock How retailers can keep up with consumers, 2013.
\newblock URL \url{https://www.mckinsey.com/industries/retail/our-insights/how-retailers-can-keep-up-with-consumers}.

\bibitem[McAuley et~al.(2015)McAuley, Targett, Shi, and van~den Hengel]{mcauley2015image}
Julian~J. McAuley, Christopher Targett, Qinfeng Shi, and Anton van~den Hengel.
\newblock Image-based recommendations on styles and substitutes.
\newblock In \emph{{SIGIR}}, pages 43--52, 2015.

\bibitem[Mehrotra et~al.(2018)Mehrotra, McInerney, Bouchard, Lalmas, and Diaz]{mehrotra2018towards}
Rishabh Mehrotra, James McInerney, Hugues Bouchard, Mounia Lalmas, and Fernando Diaz.
\newblock Towards a fair marketplace: Counterfactual evaluation of the trade-off between relevance, fairness {\&} satisfaction in recommendation systems.
\newblock In \emph{{CIKM}}, pages 2243--2251, 2018.

\bibitem[Milano et~al.(2020)Milano, Taddeo, and Floridi]{milano2020recommender}
Silvia Milano, Mariarosaria Taddeo, and Luciano Floridi.
\newblock Recommender systems and their ethical challenges.
\newblock \emph{{AI} Soc.}, 35\penalty0 (4):\penalty0 957--967, 2020.

\bibitem[Mladenov et~al.(2020)Mladenov, Creager, Ben{-}Porat, Swersky, Zemel, and Boutilier]{mladenov2020optimizing}
Martin Mladenov, Elliot Creager, Omer Ben{-}Porat, Kevin Swersky, Richard~S. Zemel, and Craig Boutilier.
\newblock Optimizing long-term social welfare in recommender systems: {A} constrained matching approach.
\newblock In \emph{Proceedings of the 37th International Conference on Machine Learning, {ICML}}, pages 6987--6998, 2020.

\bibitem[Mu{\~{n}}oz et~al.(2016)Mu{\~{n}}oz, Smith, and Patil]{executive2016big}
Cecilia Mu{\~{n}}oz, Megan Smith, and DJ~Patil.
\newblock \emph{Big data: A report on algorithmic systems, opportunity, and civil rights}.
\newblock Executive Office of the President, 2016.

\bibitem[Pariser(2011)]{pariser2011filter}
Eli Pariser.
\newblock \emph{The filter bubble: What the Internet is hiding from you}.
\newblock Penguin UK, 2011.

\bibitem[Podesta et~al.(2014)Podesta, Pritzker, Moniz, Holdren, and Zients]{united2014big}
John Podesta, Penny Pritzker, Ernest~J. Moniz, John Holdren, and Jefrey Zients.
\newblock \emph{Big data: Seizing opportunities, preserving values}.
\newblock White House, Executive Office of the President, 2014.

\bibitem[Rendle et~al.(2009)Rendle, Freudenthaler, Gantner, and Schmidt{-}Thieme]{rendle2009bpr}
Steffen Rendle, Christoph Freudenthaler, Zeno Gantner, and Lars Schmidt{-}Thieme.
\newblock {BPR:} bayesian personalized ranking from implicit feedback.
\newblock In \emph{{UAI}}, pages 452--461, 2009.

\bibitem[Sato(2022{\natexlab{a}})]{sato2022clear}
Ryoma Sato.
\newblock Clear: A fully user-side image search system.
\newblock In \emph{{CIKM}}, 2022{\natexlab{a}}.

\bibitem[Sato(2022{\natexlab{b}})]{sato2022private}
Ryoma Sato.
\newblock Private recommender systems: How can users build their own fair recommender systems without log data?
\newblock In \emph{{SDM}}, 2022{\natexlab{b}}.

\bibitem[Sato(2022{\natexlab{c}})]{sato2022retrieving}
Ryoma Sato.
\newblock Retrieving black-box optimal images from external databases.
\newblock In \emph{{WSDM}}, 2022{\natexlab{c}}.

\bibitem[Sato(2022{\natexlab{d}})]{sato2022towards}
Ryoma Sato.
\newblock Towards principled user-side recommender systems.
\newblock In \emph{{CIKM}}, pages 1757--1766, 2022{\natexlab{d}}.

\bibitem[Sato(2024{\natexlab{a}})]{sato2024making}
Ryoma Sato.
\newblock Making translators privacy-aware on the user's side.
\newblock \emph{Trans. Mach. Learn. Res.}, 2024, 2024{\natexlab{a}}.
\newblock URL \url{https://arxiv.org/abs/2312.04068}.

\bibitem[Sato(2024{\natexlab{b}})]{sato2024user}
Ryoma Sato.
\newblock User-side realization.
\newblock \emph{Doctoral Thesis}, 2024{\natexlab{b}}.
\newblock URL \url{https://arxiv.org/abs/2403.15757}.

\bibitem[Sato et~al.(2023)Sato, Takezawa, Bao, Niwa, and Yamada]{sato2023embarrassingly}
Ryoma Sato, Yuki Takezawa, Han Bao, Kenta Niwa, and Makoto Yamada.
\newblock Embarrassingly simple text watermarks.
\newblock \emph{arXiv}, abs/2310.08920, 2023.
\newblock URL \url{https://arxiv.org/abs/2310.08920}.

\bibitem[Seyerlehner et~al.(2009)Seyerlehner, Flexer, and Widmer]{seyerlehner2009limitation}
Klaus Seyerlehner, Arthur Flexer, and Gerhard Widmer.
\newblock On the limitations of browsing top-n recommender systems.
\newblock In \emph{RecSys}, pages 321--324, 2009.

\bibitem[Singh and Joachims(2018)]{singh2018fairness}
Ashudeep Singh and Thorsten Joachims.
\newblock Fairness of exposure in rankings.
\newblock In \emph{{KDD}}, pages 2219--2228, 2018.

\bibitem[Sinha and Swearingen(2002)]{sinha2002role}
Rashmi~R. Sinha and Kirsten Swearingen.
\newblock The role of transparency in recommender systems.
\newblock In \emph{{CHI}}, pages 830--831, 2002.

\bibitem[Steck(2018)]{steck2018calibrated}
Harald Steck.
\newblock Calibrated recommendations.
\newblock In Sole Pera, Michael~D. Ekstrand, Xavier Amatriain, and John O'Donovan, editors, \emph{{RecSys}}, pages 154--162, 2018.

\bibitem[Tintarev and Masthoff(2007)]{tintarev2007survey}
Nava Tintarev and Judith Masthoff.
\newblock A survey of explanations in recommender systems.
\newblock In \emph{{ICDE} Workshops}, pages 801--810, 2007.

\bibitem[Xiao et~al.(2019)Xiao, Zhao, Pan, Song, Zheng, and Yang]{xiao2019beyond}
Wenyi Xiao, Huan Zhao, Haojie Pan, Yangqiu Song, Vincent~W. Zheng, and Qiang Yang.
\newblock Beyond personalization: Social content recommendation for creator equality and consumer satisfaction.
\newblock In \emph{{KDD}}, pages 235--245, 2019.

\bibitem[Xu et~al.(2020)Xu, Cui, Kuang, Li, Zhou, Shen, and Cui]{xu2020algorithmic}
Renzhe Xu, Peng Cui, Kun Kuang, Bo~Li, Linjun Zhou, Zheyan Shen, and Wei Cui.
\newblock Algorithmic decision making with conditional fairness.
\newblock In \emph{{KDD}}, pages 2125--2135, 2020.

\bibitem[Yao and Huang(2017)]{yao2017beyond}
Sirui Yao and Bert Huang.
\newblock Beyond parity: Fairness objectives for collaborative filtering.
\newblock In \emph{{NeurIPS}}, pages 2921--2930, 2017.

\bibitem[Zafar et~al.(2017)Zafar, Valera, Gomez{-}Rodriguez, and Gummadi]{zafar2017fairness}
Muhammad~Bilal Zafar, Isabel Valera, Manuel Gomez{-}Rodriguez, and Krishna~P. Gummadi.
\newblock Fairness beyond disparate treatment {\&} disparate impact: Learning classification without disparate mistreatment.
\newblock In \emph{{WWW}}, pages 1171--1180, 2017.

\bibitem[Zehlike et~al.(2017)Zehlike, Bonchi, Castillo, Hajian, Megahed, and Baeza{-}Yates]{zehlike2017fair}
Meike Zehlike, Francesco Bonchi, Carlos Castillo, Sara Hajian, Mohamed Megahed, and Ricardo Baeza{-}Yates.
\newblock {FA*IR:} {A} fair top-k ranking algorithm.
\newblock In \emph{{CIKM}}, pages 1569--1578, 2017.

\bibitem[Zheng et~al.(2018)Zheng, Dave, Mishra, and Kumar]{zheng2018fairness}
Yong Zheng, Tanaya Dave, Neha Mishra, and Harshit Kumar.
\newblock Fairness in reciprocal recommendations: {A} speed-dating study.
\newblock In \emph{{UMAP}}, pages 29--34, 2018.

\bibitem[Ziegler et~al.(2005)Ziegler, McNee, Konstan, and Lausen]{ziegler2005improving}
Cai{-}Nicolas Ziegler, Sean~M. McNee, Joseph~A. Konstan, and Georg Lausen.
\newblock Improving recommendation lists through topic diversification.
\newblock In \emph{{WWW}}, pages 22--32, 2005.

\end{thebibliography}


\end{document}